\documentclass[aps,pra,reprint]{revtex4-2}

\usepackage{epsfig,graphics}
\usepackage{amsfonts,amssymb,amsmath} 
\usepackage{amsthm}         
\usepackage{color}
\usepackage{makecell}
\usepackage{xcolor}

\definecolor{myblue}{RGB}{0, 102, 204}

\theoremstyle{plain}
\newtheorem{theorem}{Theorem}

\begin{document}
\title{Crypto-nonlocality in arbitrarily dimensional systems}

\author{Jianqi Sheng$^1$}

\author{Dongkai Zhang$^{1, 2}$}
\email[]{zhangdk@hqu.edu.cn}

\author{Lixiang Chen$^1$}
\email[]{chenlx@xmu.edu.cn}

\affiliation{$^1$Department of Physics, Xiamen University, Xiamen 361005, China}
\affiliation{$^2$College of Information Science and Engineering, Fujian Provincial Key Laboratory of Light Propagation and Transformation, Huaqiao University, Xiamen 361021, China}

\begin{abstract}
According to Bell’s theorem, any model based on local variables cannot reproduce certain quantum correlations. A critical question is whether one could devise an alternative framework, based on nonlocal variables, to reproduce quantum correlations while adhering to fundamental principles. Leggett proposed a nonlocal model, termed crypto-nonlocality, rooted in considerations of the reality of photon polarization, but this property restricted it to being bi-dimensional. In this Letter, we extend the crypto-nonlocal model to higher dimensions and develop a framework for constructing experimentally testable Leggett-type inequalities for arbitrary dimensions. Our investigation into models that yield specific predictions exceeding those of quantum mechanics is intriguing from an information-theoretic perspective and is expected to deepen our understanding of quantum correlations. 
\end{abstract}

\maketitle

There was a long debate about whether quantum mechanics could be considered a complete theory, as it did not identify the property known as ``realism'',  which postulates that definite values can be assigned to all observable quantities \cite{EPR}. 
It is attractive if some physical theory combines quantum mechanics with a realistic picture to yield fully deterministic measurement outcomes \cite{brunner2014bell}. 
In pursuit of such a framework, Leggett introduced the concept of ``crypto-nonlocality'', which allows for a highly specific decomposition of the correlations \cite{leggett}. 
More precisely, this concept suggests that the local behavior of individual subsystems within a composite quantum system should exhibit well-defined properties, as if each subsystem were in a pure quantum state, while the correlations are governed by certain nonlocal hidden variables \cite{branciard2008testing}.

It is quite natural to compare Leggett’s crypto-nonlocality to Bell’s local causality, which are both assumptions aimed at explaining correlations between distant events \cite{bell1964einstein}. 
Notably, previous experiments on Bell nonlocality cannot fully falsify the crypto-nonlocal model, as certain correlations that violate local realism remain compatible with it \cite{giustina2015significant,shalm2015strong}. 
Inspired by Leggett’s work, several groups have conducted experimental tests of Leggett-type models, which are widely regarded as exploring frameworks that relax the assumption of locality while preserving the assumption of realism \cite{groblacher2007experimental,paterek2007experimental,branciard2007experimental,branciard2008testing}. 
The experimental violation of Leggett-type inequalities has led some to conclude that any future extension of quantum theory must relinquish certain aspects of realistic descriptions \cite{groblacher2007experimental}. 
This conclusion arises from the fact that the assumption of realism has been falsified in conjunction with nonlocality in Leggett-type tests and with locality in Bell tests. 
Another explanation for the violation of Leggett-type inequalities is that it may be fundamentally impossible to deterministically describe the individual properties of the components of a maximally entangled pair \cite{branciard2008testing}. 
The exploration of Leggett's crypto-nonlocal model has already inspired some interesting insights, particularly regarding the predictive power and completeness of quantum theory \cite{Colbeck2008,colbeck2011no,colbeck2015completeness,stuart2012experimental}. Such studies hold promise for deepening our understanding of the foundational principles of quantum mechanics. 
However, investigations into crypto-nonlocality in higher-dimensional systems remain lacking, with the underlying physical principles still poorly understood and requiring further clarification.

In this work, we extend the crypto-nonlocal model based on the assumptions underlying Leggett’s construction. We demonstrate that for arbitrarily high-dimensional systems, certain quantum correlations are incompatible with the crypto-nonlocal model, and develop a novel method for constructing experimentally testable Leggett-type inequalities. 
The argument is organized as follows. 
First, we describe and analyze the concept of crypto-nonlocality introduced by Leggett. 
Next, we extend the crypto-nonlocal model to multi-dimensional systems, preserving Leggett’s framework while reintroducing specific, well-defined individual properties for the components of entangled pairs. 
Finally, building on the information-theoretic approach developed in the work of Colbeck and Renner \cite{Colbeck2008}, we generalize it to arbitrary dimensions and apply it to the Leggett-type crypto-nonlocal model, thereby demonstrating its incompatibility with the predictions of quantum mechanics.

Pure entangled states are characterized by the fact that the properties of the pair are well-defined, but those of the individual subsystems are not. 
The intuition behind the crypto-nonlocal model is to devise an alternative no-signaling model reproducing quantum correlations, in which the individual properties are given a higher degree of prediction than quantum theory. 
In the initial formulation of the crypto-nonlocal model, Leggett imposes that the individual photons of a composite system should locally behave as if they were in a pure quantum state, and obey the spin-projection rule on the marginal probabilities, which in the photonic spin-\(1/2\) system corresponds to the well-known Malus law for polarization \cite{leggett}. 
Here, we follow and extend the framework and foundational assumptions of Leggett’s construction. 
One assumption is the no-signaling condition, which means that faster-than-light communication is impossible \cite{cirel1980quantum,popescu1994quantum}. 
The other assumption is that the physical properties of local individual particles are well-defined, such as being described by a pure quantum state. 
Under these assumptions, the Leggett-type model can replicate the predictions of quantum theory in single-partite scenarios while offering more definite predictions for the properties of individual local particles in bipartite systems. 
However, it fails to fully reproduce quantum correlations. By focusing on marginal distributions and setting aside the correlations and their underlying mechanisms, this framework highlights the contrast between nonlocal realism and quantum predictions. 
This incompatibility allows us to derive constraints, in the form of inequalities, based on the no-signaling condition and the marginal distributions. 

\begin{figure}[tb]
\centering
\includegraphics[width=8.6cm]{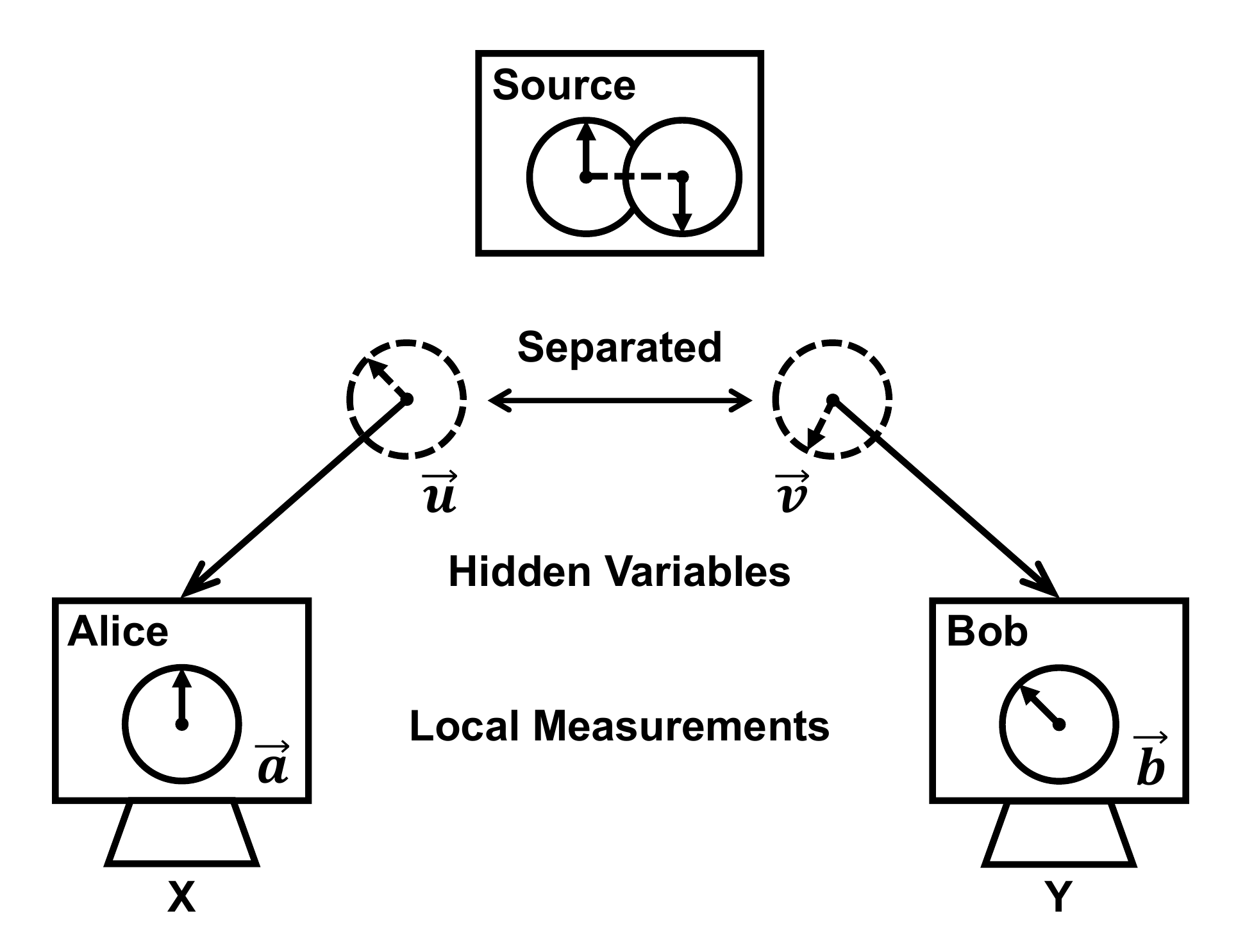}
\caption{{\sf A source emits two particles that travel to distant detectors. Under the assumptions of the crypto-nonlocal model, the physical property of each local individual particle can be defined by a pure quantum state, denoted by unit vectors \(\vec{u}\) and \(\vec{v}\). The measurement settings, freely and independently chosen by Alice and Bob, are denoted by unit vectors \(\vec{a}\) and \(\vec{b}\), with the corresponding measurement outcomes \(X\) and \(Y\), respectively. }}
\label{fig1}
\end{figure}

We focus on the scenario of a \(d\)-dimensional two-particle composite system. 
The two particles are distributed to spatially separated detectors, controlled by Alice and Bob, respectively. 
The measurement settings, freely and independently chosen by Alice and Bob, can be described by unit vectors \(\vec{a}\) and \(\vec{b}\) on the generalized Bloch sphere \(\mathbb{S}^{d^2-2}\) \cite{Bertlmann}. 
Each measurement has \(d\) possible outcomes, which we denote as \(X\) for Alice's output and \(Y\) for Bob's output. 
Physically, Alice and Bob each hold a subsystem of the shared quantum state, naturally described by their respective generalized Bloch vectors. 
Formally, the correlation observed can be described by a conditional probability distribution \(P(X,Y | \vec{a}, \vec{b})\). 
According to the hidden variable model, the outcomes of such measurements are assumed to depend on an underlying hidden random variable \(\lambda\). The conditional probability distribution is then decomposed into a statistical mixture of simpler correlations, \(P(X,Y | \vec{a}, \vec{b})=\int \mathrm{d} \lambda \rho(\lambda) P_\lambda(X,Y | \vec{a}, \vec{b})\), where \(\rho(\lambda)\) is the probability distribution of the hidden variable \(\lambda\). 
The no-signaling condition means that Alice’s local statistics are not influenced by Bob’s choice of measurement, \(P_\lambda(X| \vec{a}, \vec{b})=P_\lambda(X| \vec{a})\), and reciprocally, \(P_\lambda(Y | \vec{a}, \vec{b})=P_\lambda(Y| \vec{b})\) \cite{branciard2008testing}.

Following Leggett’s assumptions, realism is assumed to hold and corresponds to the hidden variables \(\lambda\), which are described as a product state of two qudits, \(\lambda=|\vec{u}\rangle \otimes|\vec{v}\rangle\), denoted by unit vectors \(\vec{u}\) and \(\vec{v}\). On each local side, the subsystem is assumed to have well-defined properties, described by the pure states \(|\vec{u}\rangle\) and \(|\vec{v}\rangle\). The local marginal probability distribution for Alice's side takes the form of 
\begin{align}\label{model}
	P_{\lambda}\left ( X=x| \vec{a}\right ) =\frac{1}{d}\left[1+(d-1) \vec{a^x} \cdot \vec{u}\right]\,, 
\end{align}
where the measurement settings are represented by the \((d^2-1)\)-dimensional general Bloch vectors \(\vec{a^x}\), the superscript \(x\) corresponds to the different outcomes \(X=x\). 
A similar symmetry applies to Bob’s side, where the marginal probability of detecting a particle in a specific mode must adhere to the quantum mechanical projection rule. This ensures consistency with previous experimental results while introducing no additional assumptions about the fundamental nature of the theory. 
The formulation of the local marginal probability distributions follows the standard structure predicted by quantum mechanics, consistent with experimental observations. This approach is analogous to the adherence to Malus' law in Leggett’s original model \cite{leggett}. 
Note that for a bi-dimensional system, this rule on the marginal probability distribution degenerates equivalent to Malus' law. 
Since Leggett’s assumptions focus exclusively on the local marginals without specifying the correlation, the model can still be non-local and violate a Bell inequality. 
The model reproduces the predictions of quantum theory in single-partite scenarios and makes more definite predictions for the properties of individual local particles in composite systems than quantum theory.

Our argument is motivated by the work of Colbeck and Renner \cite{Colbeck2008}, which demonstrate that quantum correlations cannot be reproduced by no-signaling theories that provide more accurate predictions of individual properties than quantum theory. 
The argument has been applied to experimentally falsify the Leggett-type hidden models in bi-dimensional systems \cite{colbeck2011no,stuart2012experimental}. 
Building on this, we extend the information-theoretic approach from Ref. \cite{Colbeck2008}, generalize it to arbitrarily high dimensions, and apply it to the Leggett-type crypto-nonlocal model, demonstrating the incompatibility between the model's predictions and those of quantum mechanics. 
Our argument proceeds in two steps. 
First, we establish a theorem that bounds the local marginals of the probability distributions using a quantity \(I_N\), introduced by Barrett \textit{et al.} in Ref \cite{barrett2006maximally}. 
Second, we show that there exist experimentally verifiable quantum correlations with sufficiently small values of \(I_N\), that are incompatible with the crypto-nonlocal model described above.

The scenario under consideration involves a bipartite setup in which a source emits two particles that are sent to two spatially separated detectors. These detectors are operated by two parties, Alice and Bob. Each party independently selects one measurement from \(N\) possible choices, with Alice’s measurement choice denoted by \(A\) and Bob’s by \(B\), where \(A, B\in \{1, \ldots, N\}\). The measurement devices then produce outcomes \(X\) and \(Y\) on Alice’s and Bob’s sides, respectively, where \(X,Y\in \{0, \ldots, d-1\}\). The measurements give rise to a joint probability distribution \(P(X,Y|\vec{a}_A,\vec{b}_ B)\) from which we quantify the correlations relevant to our statement in terms of \(I_N\), which was proposed in \cite{barrett2006maximally}, defined by, 
\begin{align}\label{chainbell}
&I_N =I_N\left \{  P(X,Y|\vec{a}_A,\vec{b}_ B)\right\}\nonumber\\
&:=\sum_{i=1}^N\left(\left\langle\left[X_i-Y_i\right]\right\rangle+\left\langle\left[Y_i-X_{i+1}\right]\right\rangle\right) ,
\end{align}
where \(\langle \cdot \rangle=\sum_{k=0}^{d-1} k P(\cdot =k)\) denotes the average value of the random variable `` \(\cdot\) '', and \(\left[\cdot\right]\) denotes `` \(\cdot\) '' modulo \(d\). Defining \(X_{N+1}:=X_{1}+1\) modulo \(d\). 
Suppose that Alice and Bob share the maximally entangled state, \(\left|\psi_d\right\rangle=\frac{1}{\sqrt{d}} \sum_{j=0}^{d-1}|j\rangle_A|j\rangle_B\), and their measurements take the form of projections onto the states \cite{collins2002bell}
\begin{align}
   & |X_A\rangle=\frac{1}{\sqrt{d}} \sum_{j=0}^{d-1} \exp \left[\frac{2 \pi i}{d} j\left(X-\alpha_A\right)\right]|j\rangle\,,\label{setA}\\
   & |Y_B\rangle=\frac{1}{\sqrt{d}} \sum_{j=0}^{d-1} \exp \left[-\frac{2 \pi i}{d} j\left(Y-\beta_B\right)\right]|j\rangle\,,\label{setB}
\end{align}
where \(\alpha_A=(A-1/2)/N\) and \(\beta_B=B/N\). With these measurement settings, quantum mechanics predicts that \(I_N=2\gamma/N+O\left(1 / N^2\right)\), denoted as \(I_N^{QM}\), where \(\gamma=\pi^2 /\left(4 d^2\right) \times\sum_{j=1}^{d-1} j / \sin ^2(\pi j / d)\). For sufficiently large \(N\), the quantity \(I_N^{QM}\) can be made arbitrarily small. 

We further consider potential alternatives to quantum theory that may admit a different description of the measurement process, as described in \eqref{model}. 
The hidden random variable \(\lambda\), described as a product state of unit vectors \(\vec{u}\) and \(\vec{v}\), might provide information leading to a joint distribution \(P(X,Y\mid \vec{a}_A,\vec{b}_ B, \vec{u},\vec{v})\). 
We now state our main theorem, which shows any no-signaling model reproducing the correlations defined by \eqref{chainbell}-\eqref{setB} has the property that the outcomes of any local measurement are completely undetermined and occur with the same probability. 
To quantify the statistical distance between two \(d\)-dimensional probability distributions \(P_X\) and \(Q_X\), we define here a quantity as \(\Delta\left(P_{X }, Q_{X }\right):= \sum_x\left|P_X(x)-Q_X(x)\right|/d\). 
\begin{theorem}\label{Theorem1}
Let \(P(X,Y\mid \vec{a}_A,\vec{b}_ B, \vec{u},\vec{v})\) be a conditional probability distribution that obeys the no-signaling condition. For any \(\vec{a}\) and \(\vec{b}\) chosen at random, we have 
\begin{align}\label{lemmaeq}
   \langle\Delta(P_{X \mid \vec{a}}, P_{X +1 \mid \vec{a}})\rangle_{\vec{u}} \leq I_N\left \{  P(X,Y|\vec{a}_A,\vec{b}_ B)\right\} ,
\end{align}
where \(\langle\cdot\rangle_{\vec{u}}\) denotes the average over the distributions of variable \(\vec{u}\). 
\end{theorem}

\begin{proof}

From the fact that all probability distributions should be non-negative, we have 
\begin{align}\label{delta1}
 \langle[\cdot]\rangle=\sum_{x=0}^{d-1} x P([\cdot]=x) \geq 1-P([\cdot]=0),
\end{align}
for any variable `` \(\cdot\) ''. We also have the following inequality under the no-signaling condition \cite{barrett2006maximally}:
\begin{align}
 \begin{aligned}
&P\left(X_A=Y_B\right)=  \sum_{x=0}^{d-1} P\left(X_A=x, Y_B=x\right) \\
&\leq  \min \left(P\left(X_A=x\right), P\left(Y_B=x\right)\right)\\
&+\min \left(1-P\left(X_A=x\right), 1-P\left(Y_B=x\right)\right)\\
&=  1-\left|P\left(X_A=x\right)-P\left(Y_B=x\right)\right|.\label{7}
\end{aligned}
\end{align}
where the function \(\min \left(P, P'\right)\) returns the smaller of the two values \(P\) and \(P'\). The inequality is valid for any certain \(x\in \{0, \ldots, d-1\}\). Combining Eq.\eqref{7} and our definition of the statistical distance \(\Delta\left(P_{X }, Q_{X }\right):= \sum_x\left|P_X(x)-Q_X(x)\right|/d\), we have that 
\begin{align}
   & d \cdot P\left(X_A=Y_B\right)\nonumber\\
   &\leq d-\sum_{x=0}^{d-1}\left|P\left(X_A=x\right)-P\left(Y_B=x\right)\right|\nonumber\\
   &=d-d\Delta\left(P_{X_A }, P_{Y_B }\right)\label{8}
\end{align}
i.e., \(P\left(X_A=Y_B\right)\leq 1-\Delta\left(P_{X_A }, P_{Y_B }\right)\). Applying the conditions derived in Eq. \eqref{delta1} to the correlations \(I_N\) given in Eq. \eqref{chainbell}, and defining \(X_{N+1}:=X_{1}+1\pmod{d}\), we obtain: 
\begin{align}
 &I_N\left \{  P(X,Y|\vec{a}_A,\vec{b}_ B, \vec{u},\vec{v})\right\}\nonumber\\
 &\geq 2 N-\sum_{i=1}^N\left[P_{ \vec{u},\vec{v}}\left(X_i=Y_i\right)+P_{ \vec{u},\vec{v}}\left(X_{i+1}=Y_i\right)\right]\nonumber\\
 &\geq \sum_{i=1}^N\left[(\Delta\left(P_{X_i|\vec{a}_i,\vec{u}}, P_{Y_i|\vec{b}_i,\vec{v}}\right)+\Delta\left(P_{X_{i+1}|\vec{a}_{i+1},\vec{u}}, P_{Y_i|\vec{b}_i,\vec{v}}\right)\right]\nonumber\\
 &\geq \sum_{i=1}^N\left[(\Delta\left(P_{X_i|\vec{a}_{i},\vec{u}}, P_{X_{i+1|\vec{a}_{i+1},\vec{u}}}\right)\right]\nonumber\\
 &\geq \Delta\left(P_{X _1|\vec{a}_{1},\vec{u}}, P_{X _{N+1}|\vec{a}_{N+1},\vec{u}}\right) \label{final}
\end{align}
where the second inequality arises from the conditions derived in Eq. \eqref{8}, and the third and fourth inequality follows from the convexity of the function: \(\Delta\left(P_{1 }, P_{2 }\right)+\Delta\left(P_{2 }, P_{3 }\right)\geq \Delta\left(P_{1 }, P_{3 }\right)\). 
For clarity, we write \(\Delta\left(P_{X _1}, P_{X _{1}+1}\right)\) to denote \( \sum_x\left|P\left(X_1=x\right)-P\left(X_1=x+1\right)\right|/d\). 
The formulation is symmetric, making the claim applicable to any \(P_{X}\), not only for \(P_{X _1}\). By taking the average over \(\vec{u}\) and \(\vec{v}\) on both sides of \eqref{final}, the theorem proved. 

\end{proof}

\begin{figure}[!t]
\centering
\includegraphics[width=8.6cm]{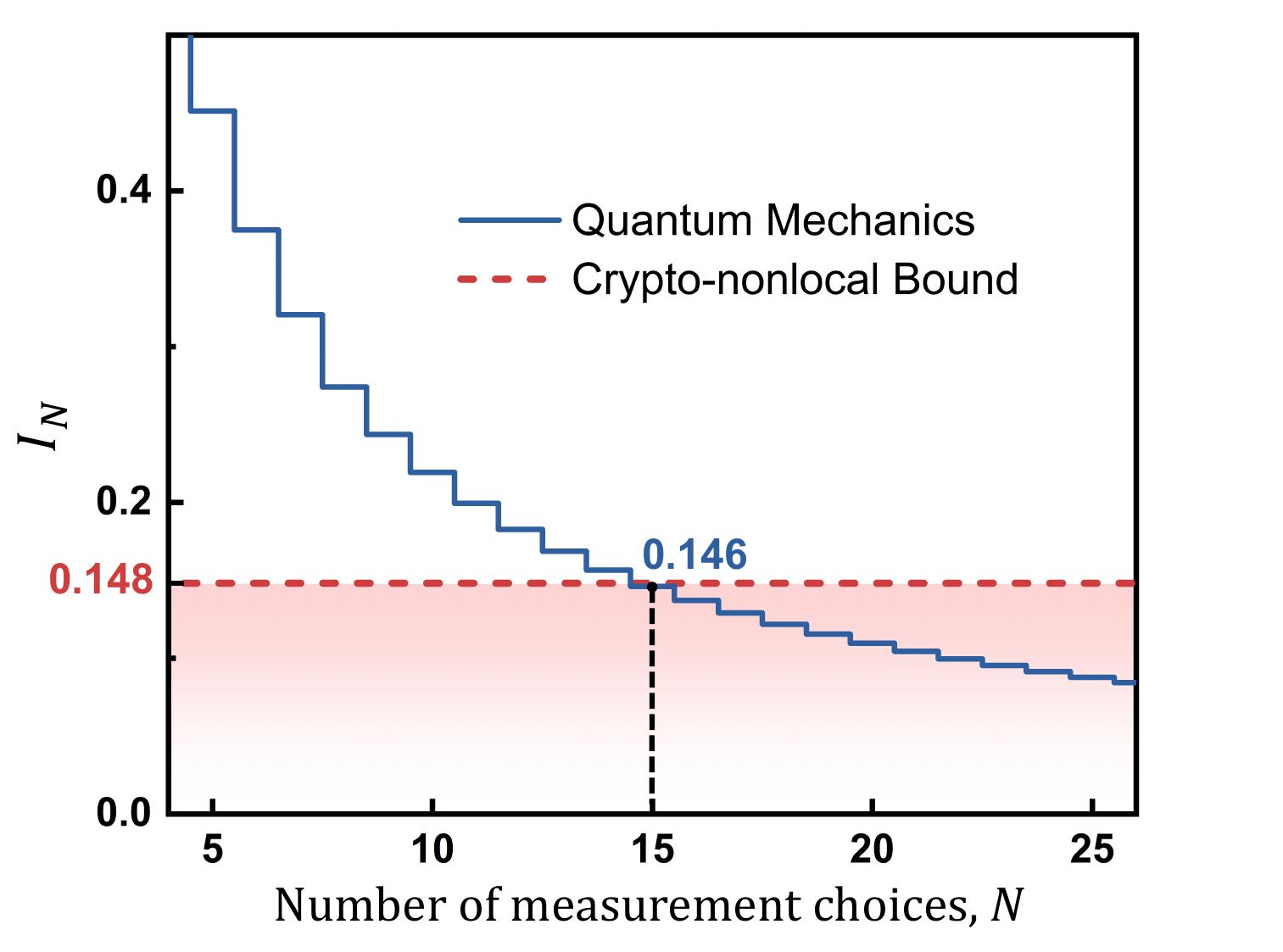}
\caption{{\sf The solid line (blue) is the theoretical curve that shows the quantum mechanics prediction for \(I_N\) with a three-dimensional maximally entangled state, \(\left|\psi_3\right\rangle=\frac{1}{\sqrt{3}} \sum_{j=0}^{2}|j\rangle_A|j\rangle_B\), versus the number of measurement bases per side used, \(N\). 
The dashed line (red) indicates the lower bound restricted by the generalized Leggett-type crypto-nonlocal model for three-dimensional systems. 
The restriction is violated for \(N\geq15\), and the model is falsified.}}
\label{fig2}
\end{figure}

For the goal of this work, we apply the theorem to the generalized Leggett-type crypto-nonlocal model described above. 
Based on Theorem 1 and the definition of the marginal probability distributions in \eqref{model}, the model imposes the constraint that 
\begin{align}\label{L}
\frac{d-1}{d^2}\sum_{x=0}^{d-1}\left \langle \left | \left( \vec{a^x}-\vec{a^{x-1}} \right)   \cdot\vec{u}  \right | \right \rangle   _{\vec{u}}\leq I_N,
\end{align}
which is similar in spirit to the Leggett inequality. For simplicity, we denote the left-hand side of the inequality as \(L\). 
Since the model is not fully specified due to the lack of an explicit distribution of the hidden variable \(\vec{u}\), we analyze the model under two possible cases, considering the implications of different distributions for \(\vec{u}\). 
First, we consider the scenario where \(\vec{u}\) is uniformly distributed over the entire Bloch sphere, which is a more natural assumption than restricting it to a specific direction corresponding to the chosen measurements. 
In this case, the inequality has an explicit lower bound, 
\(L \geq\frac{2(d-1)}{d^2}\langle | ( \vec{a^1}-\vec{a^{0}} )   \cdot\vec{u}   |  \rangle   _{\vec{u}} \geq\frac{2(d-1)}{d^3}\). 
For sufficiently large \(N\), quantum mechanics predicts \(I_{N}^{QM}\) can be made arbitrarily small, allowing the inequality to be violated. 
For example, we show the theoretically predicted violations by a three-dimensional maximally entangled state in FIG. \ref{fig2}. 
The violation of the inequality provides solid evidence that it is impossible to reconstruct quantum correlations when the properties of individual subsystems are described by a pure quantum state. Another question arises when we relax the assumption to allow the physical properties of individual particles to be only partially defined. 
Following the discussion by Branciard \textit{et al.} in Ref. \cite{branciard2008testing}, we consider a generalized model where the `local states' of the hidden variables \(\lambda\) are allowed to be mixed states with a degree of purity \(\eta \). 
Accordingly, we modify Eq. \eqref{model} to take the form 
\(P_{\lambda} ( X=x| \vec{a} ) =\frac{1}{d}[1+\eta(d-1)  \vec{a^x} \cdot \vec{u}]\), with a similar definition for Bob's side. The derivation of inequalities for this modified model follows the same way above. The constraint \eqref{L} became that \(\frac{\eta(d-1)}{d^2}\sum_{x=0}^{d-1} \langle | ( \vec{a^x}-\vec{a^{x-1}} )   \cdot\vec{u}   |  \rangle   _{\vec{u}}\leq I_N\). 
In FIG. \ref{fig3}, we give the critical number of measurement settings \(N\) per side sufficiently to falsify the Leggett-type model versus the number of dimensions \(d\), shown for degrees of purity \(\eta =1\), \(0.9\), \(0.7\), and \(0.5\). 
For large enough \(N\), the inequality is violated for any value of \(0<\eta \leq 1\). 
Thus, when the degree of purity is non-zero, a stronger statement can be concluded that individual properties cannot be even partially defined. 

\begin{figure}[!t]
\centering
\includegraphics[width=8.75cm]{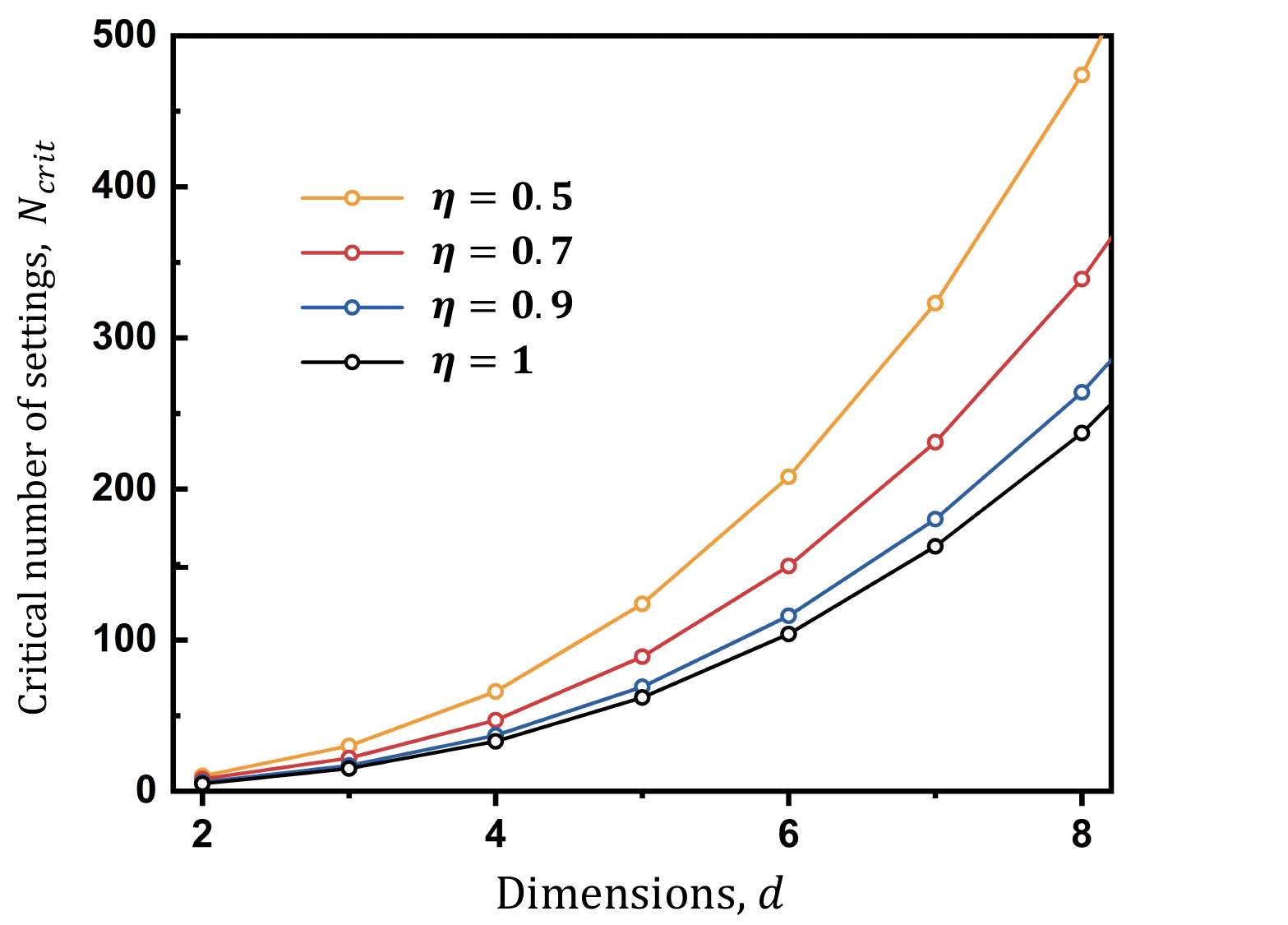}
\caption{{\sf The critical number of measurement settings \(N_{crit}\) per side sufficiently to falsify the Leggett-type model versus the number of dimensions \(d\), shown for degrees of purity \(\eta =1\), \(0.9\), \(0.7\), and \(0.5\) (from bottom to top).}}
\label{fig3}
\end{figure}

In the second scenario, \(\vec{u}\) is specified in a certain direction. 
In this case, the inequality can be violated for sufficiently large \(N\), unless \(\vec{u}\) is completely orthogonal to the plane containing all measurement vectors. 
To address this limitation, an additional set of measurements can be introduced, identical to the original set but lying in an orthogonal plane. This additional set ensures that the constraints on \(\vec{u}\) do not prevent violations of the inequality. 
Denote the measurement sets as \(S_1\) and \(S_2\) and measure the quantity \(I_N\) independently within the plane of each set, resulting in the values \(I_N^1\) and \(I_N^2\). 
The constraint in \eqref{L} became that \(L\leq \min \left ( I_N^1,I_N^2  \right )  \) for all \(\vec{a} \in S_1 \cup S_2\). 
For a quantum system with \(d\) levels, where the associated Hilbert space is \(\mathbb{H}^d \), the generalized Bloch sphere is embedded in a \(d^2-1\) dimensional Euclidean space, \(\mathbb{R}^{d^2-1} \). In this space, a unit vector can be orthogonal to a maximum of \(d^2-2\) mutually orthogonal unit vectors. 
Therefore, to ensure the inequality can be violated, a sufficient number of measurement sets is required, with at most \(d^2-2\) sets being necessary. 
As a result, quantum mechanics predict that the local hidden variables \(\vec{u}\) cannot influence the local measurement outcomes. 
By choosing appropriate measurement sets, we establish an experimentally testable Leggett-type inequality for arbitrary dimensions, which can be observed using current experimental technology \cite{dada2011experimental}.

It is natural to compare the constraints imposed by Leggett’s crypto-nonlocality and Bell’s local causality, which are both assumptions to explain correlations between distant events. 
We consider here the relation between the violation of the Leggett inequality and the Bell inequality for arbitrarily multi-dimensional systems. 
Similar to the description above, the correlations predicted by Bell’s local causality model also satisfy another inequality of the quantity, \( I_N\geq d-1\) \cite{collins2002bell}. 
When the measurement settings are under some certain restriction, such that the observed quantity satisfies \( L \leq I_N< d-1\), the correlations are compatible with the crypto-nonlocal model but can violate the local causality model. 
On the other hand, if we consider the correlations to be deterministic, meaning the probability distributions \(P(X, Y| \vec{a}, \vec{b}) \) are assigned the values 0 or 1, the correlations are local and fully predictable. 
Such correlations are compatible with Bell’s local causality model. but not always compatible with the crypto-nonlocal model. 
Consider two different settings \(\vec{a_1}\) and \(\vec{a_2}\) for Alice that \(P (X,Y| \vec{a_1}, \vec{b})=P (X,Y| \vec{a_2}, \vec{b})=1\). Following the crypto-nonlocality model must be such that \(\vec{a_1} \cdot \vec{u}=\vec{a_2} \cdot \vec{u}=1\), which is impossible for \(\vec{a_1}\ne\vec{a_2}\). 
In principle, it is possible to falsify crypto-nonlocality by considering only one party. In contrast, to falsify local causality, the measurements of two distant parties are required, as it involves correlations between events that are spatially separated. 
This difference in experimental setup implies that there is no direct logical relation between the two concepts. 
Hence, the two models represent distinct frameworks for understanding the nature of quantum correlations, and violations of one do not necessarily imply violations of the other.

In conclusion, we reformulated the crypto-nonlocal model with the constraint of the marginal distribution hypothesis, making no assumptions about the time order of events, all assumptions are made based on the local part of the correlation. 
Similar to locality, the investigation of crypto-nonlocality aims to explore non-local models that may either align with or deviate from quantum predictions, thereby providing insight into the fundamental nature of quantum correlations. 
By generalizing an information-theoretic theorem \cite{Colbeck2008} to arbitrary dimensions, we demonstrated that such a model offering more advanced predictions of individual properties than quantum theory is incompatible with quantum mechanics, which constitutes a good complement to the previous results \cite{Colbeck2008,colbeck2011no,colbeck2015completeness,stuart2012experimental}. 
The incompatibility enables us to derive constraints in the form of inequalities, offering a theoretical foundation for experimental tests of crypto-nonlocality in arbitrarily high-dimensional systems. 
The generalized Leggett-type model and its limits reported here contribute new insights into the investigation of quantum correlations and are expected to lead to a deeper understanding of the fundamental nature of quantum mechanics.

\begin{acknowledgments} 
This work was supported by National Natural Science Foundation of China (12034016, 12205107), the National Key R\&D Program of China (2023YFA1407200), Natural Science Foundation of Fujian Province of China (2021J02002) for Distinguished Young Scientists (2015J06002), Program for New Century Excellent Talents in University (NCET-13-0495), and Natural Science Foundation of Xiamen City (3502Z20227033).
\end{acknowledgments}

\bibliographystyle{apsrev4-2} 

\end{document}